\documentclass[a4paper]{article}

\usepackage{algorithm}
\usepackage[noend]{algpseudocode}

\newcommand{\sm}{\setminus}%集合差
\newcommand{\abs}[1]{\left \lvert #1 \right \rvert}%絶対値
\newcommand{\f}[2]{\frac{#1}{#2}}%分数
\newcommand{\OPT}{\mathrm{OPT}}%最適解，最適値

\newcommand{\bbr}{{\mathbb R}}

\newcommand{\R}{\bbr}%実数
\newcommand{\grs}{{\sigma}}

\newcommand{\ve}{{\varepsilon}}

\newcommand{\cala}{\mathcal{A}}

\newcommand{\calg}{\mathcal{G}}

\newcommand{\calj}{\mathcal{J}}

\newcommand{\calm}{\mathcal{M}}

\usepackage{amsmath}
\usepackage{amssymb}
\usepackage{amsfonts}
\usepackage{graphicx}
\usepackage{amsthm} 
\usepackage{tabularx}
\usepackage{fullpage}
\usepackage{environ}
\usepackage{hyperref}
\usepackage{cite}

\makeatletter
\newcommand{\problemtitle}[1]{\gdef\@problemtitle{#1}}% Store problem title
\newcommand{\probleminput}[1]{\gdef\@probleminput{#1}}% Store problem input
\newcommand{\problemoutput}[1]{\gdef\@problemoutput{#1}}% Store problem question
\NewEnviron{problem}{
  \problemtitle{}\probleminput{}\problemoutput{}% Default input is empty
  \BODY% Parse input
  \par\addvspace{.5\baselineskip}
  \noindent{
    \framebox[\textwidth][c]{
        \begin{tabularx}{\textwidth}{@{\hspace{\parindent}} l X c}
            \multicolumn{2}{@{\hspace{\parindent}}l}{\textsc{\@problemtitle}} \\% Title
            \textbf{Input:} & \@probleminput \\% Input
            \textbf{Output:} & \@problemoutput% Output
        \end{tabularx}
        }
      }
  \par\addvspace{.5\baselineskip}
}
\makeatother

\newtheorem{theorem}{Theorem}
\newtheorem{lemma}[theorem]{Lemma}
\newtheorem{proposition}[theorem]{Proposition}

\newtheorem{definition}[theorem]{Definition}

\title{An Additive Approximation Scheme for the Nash Social Welfare Maximization with Identical Additive Valuations\thanks{This work is partly supported by JSPS, KAKENHI grant numbers JP18H05291, JP19H05485, and JP20K11692, Japan.}} 

\author{Asei Inoue\thanks{Kyoto University, Japan.} \and Yusuke Kobayashi\thanks{Kyoto University, Japan. yusuke@kurims.kyoto-u.ac.jp}}

\begin{document}

\maketitle

\begin{abstract}
We study the problem of efficiently and fairly allocating a set of indivisible goods among agents with identical and additive valuations for the goods. The objective is to maximize the Nash social welfare, which is the geometric mean of the agents' valuations. While maximizing the Nash social welfare is NP-hard, a PTAS for this problem is presented by Nguyen and Rothe. The main contribution of this paper is to design a first additive PTAS for this problem, that is, we give a polynomial-time algorithm that maximizes the Nash social welfare within an additive error $\varepsilon v_{\rm max}$, where $\varepsilon$ is an arbitrary positive number and $v_{\rm max}$ is the maximum utility of a good. The approximation performance of our algorithm is better than that of a PTAS. The idea of our algorithm is simple; we apply a preprocessing and then utilize an additive PTAS for the target load balancing problem given recently by Buchem et al. However, a nontrivial amount of work is required to evaluate the additive error of the output. 
\end{abstract}

\section{Introduction}

\subsection{Nash Social Welfare Maximization}

We study the problem of efficiently and fairly allocating a set of indivisible goods among agents with identical and additive valuations for the goods. There are many ways to measure the quality of the allocation in the literature, and in this paper, we aim to maximize the Nash social welfare \cite{kaneko1979nash}, which is the geometric mean of the agents' valuations in the allocation. 

Suppose we are given a set of agents $\cala = \{1,2 ,\dots , n \}$ and a set of goods $\calg = \{1, 2, \dots , m \}$ with a utility $v_j > 0$ for each $j \in \calg$. An \textit{allocation} is a partition $\pi = (\pi_1, \dots , \pi_n)$ of $\calg$ where $ \pi_i  \subseteq \calg$ is a set of goods assigned to agent $i$. For an allocation $\pi = (\pi_1, \dots , \pi_n)$, let $v(\pi_i)$ be the valuation of $i$ that is defined as the sum of the utility of the goods assigned to $i$, i.e., $v(\pi_i)  = \sum_{j \in \pi_i } v_j$. 
The goal is to find an allocation $\pi$ that maximizes the function 
\[
f(\pi) = \left( \prod_{i \in \cala} v(\pi_i) \right)^{1/n},
\] 
which is called the \textit{Nash social welfare} \cite{barman2018greedy,nguyen2014minimizing}. 
In this paper, we refer to this problem as \textsc{Identical Additive NSW}.

\begin{problem}
\problemtitle{\textsc{Identical Additive NSW}}
\probleminput{A set of agents $\cala = \{1,2 ,\dots , n \}$ and a set of goods $\calg = \{1, 2, \dots , m \}$ with a utility $v_j > 0$ for each $j \in \calg$.}
\problemoutput{An allocation $\pi$ that maximizes the Nash social welfare $f(\pi)$.}
\end{problem}

The Nash social welfare can be defined in a more general setting where the valuation of each agent $i$ is determined by a set function $v_i \colon 2^{\calg}\to \R_{\geq 0}$. In such a case, the Nash social welfare of an allocation $\pi = (\pi_1, \dots , \pi_n)$ is defined as 
%%the geometric mean of the agents' valuations, that is,  
$\left( \prod_{i \in \cala} v_i(\pi_i) \right)^{1/n}$. 
In \textsc{Identical Additive NSW}, we focus on the case where the valuation function is additive and independent of the agent. Note that, by removing goods with zero utility, we can assume that $v_j >0$ without loss of generality.

The Nash social welfare was named after John Nash, who introduced and studied the Nash social welfare in the context of bargaining in the 1950s \cite{nash1950bargaining}. Later, the same concept was independently studied in the context of competitive equilibria with equal incomes \cite{varian1974equity} and proportional fairness in networking \cite{kelly1997charging}. It has traditionally been studied in the economics literature for divisible goods \cite{moulin2003fair}. For divisible goods, an allocation maximizing the Nash social welfare can be computed in polynomial time when the valuation functions are additive \cite{eisenberg1959consensus}.

In the context of goods allocation, the Nash social welfare is a measure that captures efficiency and fairness at the same time. 
To see this, for a parameter $q \in \R$ and for an allocation $\pi$, 
one can define the {\em generalized mean} of the valuation of each agent as $f_q(\pi) = \left( \f{1}{n} \sum_{i=1}^n v_i(\pi_i)^q \right)^{1/q}$. 
The generalized mean can be a variety of mean functions depending on the value of $q$. 
When $q=1$, $f_q(\pi)$ is the average valuation of the agents, and hence maximizing $f_q(\pi)$ is equivalent to maximizing the social welfare. 
In this case, $f_q(\pi)$ is a measure of the efficiency of an allocation. 
When $q \to - \infty$, $f_q(\pi)$ is the minimum value of $v_i (\pi_i)$, namely the valuation of the least satisfied agent. 
In this case, an allocation maximizing $f_q(\pi)$ can be considered fair in a sense. 
It is known that in the limit as $q \to 0$, $f_q(\pi)$ coincides with the geometric mean, which is the Nash social welfare (see~\cite{bullen2013handbook}). 
Therefore, maximizing the Nash social welfare (i.e., $q \to 0$) can be viewed as a compromise between Maximum Social Welfare (i.e., $q=1$) and Max-Min Welfare (i.e., $q \to -\infty$).

The Nash social welfare is closely related to other concepts EF1 and Pareto optimality that describe fairness and efficiency, respectively,  
which also supports the importance of the Nash social welfare. 
An allocation is said to be {\em EF1} ({\em envy-free up to at most one good}) if each agent prefers its own bundle over the bundle of any other agent up to the removal of one good. 
An allocation is called {\em Pareto optimal} if no one else's valuation can be increased without sacrificing someone else's valuation. 
Caragiannis et al.~\cite{caragiannis2019unreasonable} showed that an allocation that maximizes the Nash social welfare is both EF1 and Pareto optimal when agents have additive valuations for the goods. 
This motivates studying the problem of finding an allocation that maximizes the Nash social welfare.

\subsection{Our Contribution: Approximation Algorithm}

The topic of this paper is the approximability of the Nash social welfare maximization. 
By an easy reduction from the Subset Sum problem, we can see that maximizing the Nash social welfare is NP-hard even in the case of two agents with identical additive valuations. That is, \textsc{Identical Additive NSW} is NP-hard even when $n=2$. Furthermore, maximizing the Nash social welfare is APX-hard for multiple agents with non-identical valuations even when the valuations are additive \cite{lee2017apx}. 

On a positive side, several approximation algorithms are proposed for maximizing the Nash social welfare, and the difficulty of the problem depends on the class of valuations $v_i$. Under the assumption that the valuation set function is monotone and submodular, 
Li and Vondr{\'a}k~\cite{li2021constant} recently proposed
a constant factor approximation algorithm 
%%, and its approximation ratio is $380$ \cite{li2021constant}. 
based on an algorithm for Rado valuations~\cite{10.1145/3406325.3451031}. 
Better constant factor approximation algorithms are known for subclasses of submodular functions~\cite{10.5555/3458064.3458133,10.5555/3174304.3175452,chaudhury_et_al:LIPIcs:2018:9924,DBLP:conf/soda/GargHM18}. 
When the valuation function is additive, a $1.45$-approximation algorithm is known \cite{barman2018finding}, and this is the current best approximation ratio. When the valuation functions are additive and identical, the situation is much more tractable. Indeed, for \textsc{Identical Additive NSW}, it is known that a polynomial-time approximation scheme (PTAS) exists \cite{nguyen2014minimizing} and a simple fast greedy algorithm achieves a $1.061$-approximation guarantee \cite{barman2018greedy}. 

For \textsc{Identical Additive NSW}, 
the above results show the limit of the approximability and so no further improvement seems to be possible in terms of the approximation ratio. Nevertheless, a better approximation algorithm may exist if we evaluate the approximation performance in a fine-grained way. The main contribution of this paper is to show that this is indeed the case if we evaluate the approximation performance by using the additive error. Formally, our result is stated as follows.

\begin{theorem}\label{main result}
For an instance of \textsc{Identical Additive NSW}, let $v_{\max} = \max_{j \in \calg} v_j$ and let $\OPT$ be the optimal value. 
For any $ \ve>0$, there is an algorithm $A_{\ve}$ for \textsc{Identical Additive NSW} that runs in $(n m/\ve)^{O(1/\ve)}$ time and returns an allocation $\pi$ such that 
$f(\pi) \geq \OPT - \ve v_{\max}$. 
\end{theorem}

Recall that a PTAS for \textsc{Identical Additive NSW} is an algorithm that returns an allocation $\pi$ with $f(\pi) \geq \f{\OPT}{1 + \ve}$. Since $\f{\OPT}{1 + \ve} \approx (1-\ve)\OPT$, the additive error of a PTAS is roughly $\ve \OPT$, which can be much greater than $\ve v_{\max}$. Furthermore, as we will see in Proposition~\ref{prop:ratio}, our algorithm given in the proof of Theorem~\ref{main result} is also a PTAS. 
%%our algorithm returns a better solution than a PTAS when $\OPT$ is much greater than $v_{\max}$. 
In this sense, we can say that our algorithm is better than a PTAS if we evaluate the approximation performance in a fine-grained way.  

We also note that there is no polynomial-time algorithm for
finding an allocation $\pi$ with $f(\pi) \geq \OPT - \ve$ unless ${\rm P}={\rm NP}$. 
This is because the additive error can be arbitrarily large by scaling the utility unless we obtain an optimal solution. 
Therefore, parameter $v_{\max}$ is necessary to make the condition scale-invariant.

\subsection{Related Work: Additive PTAS}

The algorithm in Theorem~\ref{main result} is called an \textit{additive PTAS} with parameter $v_{\max}$, and so our result has a meaning in a sense that it provides a new example of a problem for which an additive PTAS exists. In this subsection, we describe known results on additive PTASs, some of which are used in our argument later. 

An additive PTAS is a framework for approximation guarantees that was recently introduced by Buchem et al.~\cite{buchem2020additive,buchem_et_al:LIPIcs.ICALP.2021.42}. For any $\ve > 0$, an additive PTAS returns a solution whose additive error is at most $\ve$ times a certain parameter.

\begin{definition}
For an optimization problem, an {\em additive PTAS} is a family of polynomial-time algorithms $\{A_\ve \mid \ve > 0\}$ with the following condition: for any instance $I$ and for every $\ve > 0$, $A_\ve$ finds a solution with value $A_{\ve}(I)$ satisfying $\abs{A_{\ve}(I) - \OPT(I)} \leq \ve h$, where $h$ is a suitably chosen parameter of instance $I$ and $\OPT(I)$ is the optimal value.
\end{definition}

In some cases, an additive PTAS is immediately derived from an already known algorithm. For example, by setting the error factor appropriately, a fully polynomial-time approximation scheme (FPTAS) for the knapsack problem \cite{ibarra1975fast} is also an additive PTAS where the parameter is the maximum utility of a good. However, evaluating the additive error is difficult in general, and so additive PTASs are known for only a few problems. In the pioneering paper on additive PTASs by Buchem et al.~\cite{buchem2020additive,buchem_et_al:LIPIcs.ICALP.2021.42}, an additive PTAS was proposed for the completion time minimization scheduling problem, the Santa Claus problem, and the envy minimization problem. In order to derive these additive PTASs, they introduced the \textit{target load balancing problem} and showed that it is possible to determine whether a solution exists by only slightly violating the constraints. 

In the \textit{target load balancing problem}, we are given a set of jobs $\calj$ with a processing time $v_j > 0$ for each $j \in \calj$ and a set of machines $\calm$ with real values $l_i$ and $u_i$ for each $i \in \calm$. The goal is to assign each job $j \in \calj$ to a machine $i \in \calm$ such that for each machine $i \in \calm$ the load of $i$ (i.e., the sum of the processing times of the jobs assigned to $i$) is in the interval $[l_i,u_i]$.  In a similar way to \textsc{Identical Additive NSW}, an assignment is represented by a partition $\pi = (\pi_i)_{i \in \calm}$ of $\calj$. Let $v_{\max} = \max_{j \in \calj} v_j$ and let $K$ denote the number of types of machines, that is, $K = |\{ (l_i,u_i) \mid i \in \calm \}|$. While the target load balancing problem is NP-hard, Buchem et al.~\cite{buchem2020additive,buchem_et_al:LIPIcs.ICALP.2021.42} showed that it can be solved in polynomial time if we allow a small additive error and $K$ is a constant.

\begin{theorem}[Buchem et al.~{\cite[Theorem 12]{buchem2020additive}}] \label{Buchem}
For the target load balancing problem and for any $\ve >0$, there is an algorithm (called {\sc LoadBalancing}) that  either
\begin{enumerate}
 \item concludes that there is no feasible solution for a given instance, or
 \item returns an assignment $\pi = (\pi_i)_{i \in \calm}$ such that the total load $\sum_{j \in \pi_i} v_j$ is in $[l_i - \ve  v_{\max}, u_i + \ve v_{\max}]$ for each $i \in \calm$ 
\end{enumerate}
in $|\calm|^{K+1} (\f{|\calj|}{\ve})^{O(1/ \ve)}$ time. 
\end{theorem}

Note that the algorithm in this theorem is used as a subroutine in our additive PTAS for {\sc Identical Additive NSW}. 
Note also that the term ``assignment'' is used in this theorem by following the convention, but it just means a partition of the jobs. 
Therefore, we do not distinguish ``assignment'' and ``allocation'' in what follows in this paper.

\subsection{Technical Highlights}
\label{sec:highlights}

In this subsection, we describe the outline of our algorithm for {\sc Identical Additive NSW} and explain two technical issues that are peculiar to additive errors. 

The basic strategy of our algorithm is simple; we guess the valuation $v(\pi^*_i)$ of each agent $i$ in an optimal solution $\pi^*$, and then seek for an allocation $\pi$ such that $|v(\pi_i) - v(\pi^*_i)| \le \ve  v_{\max}$ for each $i \in \cala$ by using {\sc LoadBalancing} in Theorem~\ref{Buchem}. 

The first technical issue is that even if the additive error of $v(\pi_i)$ is at most $\ve  v_{\max}$ for each $i \in \cala$, the additive error of $f(\pi)$ is not easily bounded by $\ve  v_{\max}$. This is in contrast to the case of the multiplicative error (i.e., if $v(\pi_i) \ge v(\pi^*_i) / (1+\ve)$ for each $i \in \cala$, then $f(\pi) \ge f(\pi^*) / (1+\ve)$). The first technical ingredient in our proof is to bound the additive error of $f(\pi)$ under the assumption that $v_{\max}$ is at most the average valuation of the agents; see Lemma~\ref{preadditive} for a formal statement. In order to apply this argument, we modify a given instance so that $v_{\max}$ is at most the average valuation of the agents by a naive preprocessing. In the preprocessing, we assign a good $j$ with high utility to an arbitrary agent $i$ and remove $i$ and $j$ from the instance, repeatedly (see Section~\ref{sec:preprocessing} for details).

The second technical issue is that the preprocessing might affect the additive error of the output, whereas it does not affect the optimal solutions of the instance (see Lemma~\ref{preprocessing}). Suppose that an instance $I$ is converted to an instance $I'$ by the preprocessing, and suppose also that we obtain an allocation $\pi'$ for $I'$. Then, by recovering the agents and the goods removed in the preprocessing, we obtain an allocation $\pi$ for $I$ from $\pi'$. The issue is that the additive error of the objective function value might be amplified by this recovering process, which makes the evaluation of the additive error of $f(\pi)$ hard. Nevertheless, we show that the additive error of $f(\pi)$ is bounded by $O(\ve  v_{\max})$ with the aid of the mean-value theorem for differentiable functions (see Proposition~\ref{prop:performanceMNW}), which is the second technical ingredient in our proof. It is worth noting that the differential calculus plays a crucial role in the proof, whereas the problem setting is purely combinatorial.  

The remaining of this paper is organized as follows. In Section~\ref{sec:description}, we describe our algorithm for {\sc Identical Additive NSW}. Then, in Section~\ref{sec:analysis}, we show its approximation guarantee and prove Theorem~\ref{main result}.

\section{Description of the Algorithm}
\label{sec:description}

As we mentioned in Section~\ref{sec:highlights}, in our algorithm, we first apply a preprocessing so that $v_{\max}$ is at most the average valuation of the agents. Then, we guess the valuation of each agent in an optimal solution, and then seek for an allocation that is close to the optimal solution by using {\sc LoadBalancing}, which is the main procedure. We describe the preprocessing and the main procedure in Sections~\ref{sec:preprocessing} and \ref{sec:mainproc}, respectively. 

\subsection{Preprocessing}
\label{sec:preprocessing}

Consider an instance $I = (\cala , \calg, \mathbf{v})$ of {\sc Identical Additive NSW} where $\mathbf{v} =( v_1, \dots , v_m )$. 
%%Observe that the optimal value is positive if and only if $|\cala| \le |\calg|$. 
If $|\cala| > |\calg|$, then the optimal value is zero, and hence any allocation is optimal. 
Therefore, we may assume that $|\cala| \le |\calg|$, which implies that the optimal value is positive. 
Let $\mu(I)$ be the average valuation of agents, that is, $\mu(I) = \f{1}{| \cala |} \sum_{j \in \calg} v_j$. When $I$ is obvious, we simply write $\mu$ for $\mu(I)$.
The objective of the preprocessing is to modify a given instance so that $v_j < \mu$ for any $j \in \calg$. 

Our preprocessing immediately follows from the fact that an agent who receives a valuable good does not receive other goods in an optimal solution. 
Although similar observations were shown in previous papers (see e.g.~\cite{alon1998approximation,nguyen2014minimizing}), we give a proof for completeness. 

\begin{lemma}\label{preprocessing}
Let $j \in \calg$ be an item with $v_j \geq \mu$. In an optimal solution $\pi^*$, an agent who receives $j$ cannot receive any goods other than $j$. 
\end{lemma}

\begin{proof}
Assume to the contrary that $\pi^*$ is an optimal solution which assigns goods $j$ and $l$ to the same agent $i \in \cala$.
This implies $v( \pi^*_i ) > \mu$ and so there must be at least one agent $k\in \cala$ such that $v( \pi^*_k ) < \mu$. Let $\pi'$ be the new allocation obtained from $\pi^*$ by reassigning good $l$ to agent $k$. Then, we obtain
\begin{align*}
    v( \pi'_i ) v( \pi'_k ) - v( \pi^*_i ) v( \pi^*_k ) &= (v( \pi^*_i ) - v_l)(v( \pi^*_k )+ v_l )- v( \pi^*_i ) v( \pi^*_k ) \\
    &= v_l(v( \pi^*_i ) - v( \pi^*_k ) -v_l  ) \\
    &\geq v_l(v_j - v( \pi^*_k ) ) \\
    &\geq v_l( \mu - v( \pi^*_k ) ) \\
    &> 0.
\end{align*}
This implies that $f(\pi') > f(\pi^*)$, which contradicts the optimality of $\pi^*$. 
\end{proof}

For $\cala_0 \subseteq \cala$ and $\calg_0 \subseteq \calg$, let $I \sm (\cala_0,\calg_0)$ denote the instance obtained from $I$ by removing $\cala_0$ and $\calg_0$, that is, $I \sm (\cala_0,\calg_0) = (\cala \sm \cala_0, \calg \sm \calg_0, \mathbf{v} \sm \calg_0)$, where $\mathbf{v} \sm \calg_0 = (v_j)_{j \in \calg \sm \calg_0}$.
In the preprocessing, we assign a good $j \in \calg$ with $v_j \geq \mu$ to some agent $i \in \cala$ and remove $i$ and $j$ from the instance, repeatedly. A formal description is shown in Algorithm~\ref{alg:01}.

\begin{algorithm}[H]
\caption{ {\sc Preprocessing}}
\label{alg:01}
\begin{algorithmic}[1]
\Require instance $I = ( \cala , \calg , \mathbf{v} )$ where $\mathbf{v} = ( v_1, v_2 , \dots , v_m )$
\Ensure subsets $\cala_0 \subseteq \cala$ and $\calg_0 \subseteq \calg$
\State Initialize $\cala_0$ and $\calg_0$ as $\cala_0 = \calg_0 = \emptyset$.
\While{there exists a good $j \in \calg \sm \calg_0$ with $v_j \geq \mu(I \sm (\cala_0,\calg_0))$}
    \State $\calg_0 \leftarrow \calg_0 \cup \{j\}$
    \State Choose $i \in \cala \sm \cala_0$ arbitrarily and add it to $\cala_0$.
\EndWhile
\State \Return $\cala_0$, $\calg_0$
\end{algorithmic}
\end{algorithm}

Let $\cala_0$ and $\calg_0$ be the output of {\sc Preprocessing}. 
Lemma~\ref{preprocessing} shows that if we obtain an optimal solution for $I \sm (\cala_0,\calg_0)$, then 
we can immediately obtain an optimal solution for $I$ by assigning each good in $\calg_0$ to each agent in $\cala_0$. 
Note that the inequality $v_j < \mu(I \sm (\cala_0,\calg_0))$ holds for all $j \in \calg \sm \calg_0$ after the preprocessing. 
Thus, the maximum utility of a good is less than the average valuation of the agents in the instance $I \sm (\cala_0,\calg_0)$. 
Note also that, since the number of while loop iterations is at most $|\cala|$, {\sc Preprocessing} runs in polynomial time.

\subsection{Main Procedure}
\label{sec:mainproc}

We describe the main part of the algorithm, in which we guess the valuation of each agent in an optimal solution and then apply {\sc LoadBalancing}. 
In order to obtain a polynomial-time algorithm, we have the following difficulties: 
the number of guesses has to be bounded by a polynomial and 
the number of machine types $K$ has to be a constant when we apply {\sc LoadBalancing}. 
To overcome these difficulties, we get good upper and lower bounds on the valuation of each agent in an optimal solution, which is a key observation in our algorithm. 
We prove the following lemma by tracing the proof of Lemma~\ref{preprocessing}.

\begin{lemma} \label{limit the range}
For any instance $I$ of {\sc Identical Additive NSW}, let $\pi^*$ be an optimal allocation of $I$. Then, 
\[
\mu - v_{\max} < v( \pi^*_i ) < \mu + v_{\max}
\]
holds for any $i \in \cala$. 
\end{lemma}

\begin{proof}

Assume that there is some agent $i \in \cala$ with $v( \pi^*_i ) \geq \mu + v_{\max}$. This implies that $v( \pi^*_i ) > \mu$ and so there must be at least one agent $k \in \cala$ such that $v( \pi^*_k ) < \mu$. Let $j \in \calg$ be a good assigned to $i$. By reassigning $j$ to agent $k$, we get a new allocation $\pi'$ from $\pi^*$. Then, we obtain
\begin{align*}
    v( \pi'_i ) v( \pi'_k ) - v( \pi^*_i ) v( \pi^*_k ) &= (v( \pi^*_i ) - v_j)(v( \pi^*_k )+ v_j )- v( \pi^*_i ) v( \pi^*_k ) \\
    &= v_j(v( \pi^*_i ) - v( \pi^*_k ) -v_j  ) \\
    &> v_j(v( \pi^*_i ) - \mu - v_{\max}  ) \\
    &\geq 0.
\end{align*}
This shows that $f(\pi') > f(\pi^*)$, which contradicts the optimality of $\pi^*$. 

Assume that there is some agent $i \in \cala$ with $v( \pi^*_i ) \leq \mu - v_{\max}$. This implies that $v( \pi^*_i ) < \mu$, and so there must be at least one agent $k \in \cala$ such that $v( \pi^*_k ) > \mu$. Let $j \in \calg$ be a good assigned to $k$. Then, 
it holds that
\[
    v_j + v( \pi^*_i ) \leq v_{\max} + v( \pi^*_i ) \leq \mu < v( \pi^*_k ).
\]
By reassigning $j$ to agent $i$, we get a new allocation $\pi'$ from $\pi^*$. Then, we obtain
\begin{align*}
    v( \pi'_i ) v( \pi'_k ) - v( \pi^*_i ) v( \pi^*_k ) &= (v( \pi^*_i ) + v_j)(v( \pi^*_k ) - v_j )- v( \pi^*_i ) v( \pi^*_k ) \\
    &= v_j(v( \pi^*_k ) - v( \pi^*_i ) -v_j  ) \\
    &> 0.
\end{align*}
This shows that $f(\pi') > f(\pi^*)$, which contradicts the optimality of $\pi^*$.
\end{proof}

We are now ready to describe our algorithm. Suppose we are given an instance $I = (\cala , \calg, \mathbf{v})$ with $v_{\max} < \mu$. 
To simplify the description, suppose that $1/\ve$ is an integer. 

Our idea is to guess $v(\pi^*_i)$ with an additive error $\ve v_{\max}$ for each $i \in \cala$, where $\pi^*$ is an optimal solution.
By Lemma \ref{limit the range}, we already know that the value of an optimal solution is in the interval of width $2 v_{\max}$. Let $L$ be the set of points delimiting this interval with width $\ve v_{\max}$, that is,
\[
 L = \{ \mu - v_{\max} + i  \ve  v_{\max} \mid i \in \{0, 1, 2, \dots , 2/\ve -1\} \}.    
\]
Let $L^{\cala}$ be the set of all the maps from $\cala$ to $L$. 
For $\tau, \tau' \in L^{\cala}$, we denote $\tau \sim \tau'$ if $\tau'$ is obtained from $\tau$ by changing the roles of the agents, or equivalently 
$|\{i \in \cala \mid \tau(i)=x\}| = |\{i \in \cala \mid \tau'(i)=x\}|$ for each $x \in L$. 
In such a case, since each agent is identical, we can identify $\tau$ and $\tau'$. 
This motivates us to define $D := L^{\cala} / \sim$, where $\sim$ is the equivalence relation defined as above. 
%%Then, each $\tau\in D$ is determined by the number of agents $i\in \cala$ such that $\tau(i) = x$ for $x\in L$.

For each $\tau\in D$, we apply {\sc LoadBalancing} in Theorem~\ref{Buchem} to the following instance of the target load balancing problem: 
$\calm := \cala$, $\calj := \calg$, the processing time of $j \in \calj$ is $v_j$, and
the target interval is $[\tau(i), \tau(i) + \ve v_{\max}]$ for each $i \in \calm$. 
Then, {\sc LoadBalancing} either concludes that no solution exists or returns an assignment (allocation) $\pi^\tau$ 
such that $v(\pi^\tau_i) \in [\tau(i) - \ve v_{\max}, \tau(i) + 2 \ve v_{\max}]$ for each $i \in \calm$. 

Among all solutions $\pi^\tau$ returned by {\sc LoadBalancing}, our algorithm chooses an allocation with the largest objective function value. 
A pseudocode of our algorithm is shown in Algorithm~\ref{alg:02}.

\begin{algorithm}[H]
\caption{{ \sc MainProcedure}}
\label{alg:02}
\begin{algorithmic}[1]
\Require instance $I = (\cala , \calg, \mathbf{v})$ such that $v_{\max} < \mu$
\Ensure allocation $\pi$
\State Initialize $\pi$ as an arbitrary allocation.
\For{$\tau \in D$}
%%    \State Apply {\sc LoadBalancing} to check the existence of allocation $\pi^{\tau}$ such that $ v( \pi^{\tau}_i ) \in [\tau(i), \tau(i) + \ve v_{\max} ]$ for all $i \in \cala$.
    \State Apply {\sc LoadBalancing} with the target interval $[\tau(i), \tau(i) + \ve v_{\max} ]$ for $i \in \cala$.
    \If{{\sc LoadBalancing} returns an allocation $\pi^{\tau}$}
        \If{ $f(\pi) < f(\pi^{\tau})$}
            \State $\pi \leftarrow \pi^{\tau}$
        \EndIf
    \EndIf
\EndFor
\State \Return $\pi$
\end{algorithmic}
\end{algorithm}

\begin{proposition}\label{runningtime}
The running time of {\sc MainProcedure} is $(nm/\ve)^{O(1/\ve)} $.
\end{proposition}

\begin{proof}
To obtain an upper bound on the number of for loop iterations, we estimate the number of elements in $D$. 
Since each $\tau\in D$ is determined by the number of agents $i\in \cala$ such that $\tau(i) = x$ for $x\in L$, we obtain $|D| \leq |\{0, 1, \dots , n\}|^L \leq (n+1)^{2/\ve} = n^{O(1/\ve)}$.

We next estimate the running time of {\sc LoadBalancing}. 
Since $|\calm| = n$, $|\calj| = m$, and the number of machine types $K$ is at most $|L| =  2/\ve $, the running time of {\sc LoadBalancing} is $n^{2/\ve + 1} (m/\ve)^{O(1/\ve)}$ by Theorem~\ref{Buchem}.

Thus, the total running time of {\sc MainProcedure} is $(nm/\ve)^{O(1/\ve)}$.
\end{proof}

The entire algorithm for {\sc Identical Additive NSW} consists of the following steps: 
apply {\sc Preprocessing}, apply {\sc MainProcedure}, and recover the removed sets. 
A pseudocode of the entire algorithm is shown in Algorithm~\ref{alg:03}.

\begin{algorithm}[H]
\caption{{ \sc MaxNashWelfare} }
\label{alg:03}
\begin{algorithmic}[1]
\Require instance $I =(\cala, \calg, \mathbf{v})$ 
\Ensure allocation $\pi'$
\State Apply {\sc Preprocessing} to $I$ and obtain $\cala_0$ and $\calg_0$.
\State Apply {\sc MainProcedure} to $I \sm (\cala_0,\calg_0)$ and obtain $\pi$.
\State Let $\grs$ be a bijection from $\cala_0$ to $\calg_0$.
\State Set $\pi'_i = \{ \grs(i) \}$ for $i \in \cala_0$ and set $\pi'_i = \pi_i$ for $i \in \cala \sm \cala_0$.
\State \Return $\pi'$
\end{algorithmic}
\end{algorithm}

Since the most time consuming part is {\sc MainProcedure},  
the running time of {\sc MaxNashWelfare} is $(nm/\ve)^{O(1/\ve)}$
by Proposition~\ref{runningtime}.

\section{Analysis of Approximation Performance}
\label{sec:analysis}

In this section, we show that {\sc MaxNashWelfare} returns a good approximate solution for {\sc Identical Additive NSW} and give a proof of Theorem~\ref{main result}. 
We first analyze the performance of {\sc MainProcedure} in Section~\ref{sec:analysismain}, and then  
analyze the effect of {\sc Preprocessing} in Section~\ref{sec:analysispre}.

\subsection{Approximation Performance of {\sc MainProcedure}}
\label{sec:analysismain}

In this subsection, we consider an instance $I = (\cala , \calg, \mathbf{v})$ of {\sc Identical Additive NSW} such that $v_{\max} < \mu$. 
The following lemma is easy, but useful in our analysis of {\sc MainProcedure}.

\begin{lemma}\label{neighbor}
Assume that $v_{\max} < \mu$. Let $\pi$ be the allocation returned by {\sc MainProcedure}. For any optimal solution $\pi^*$, there exists an allocation $\pi^{\tau}$ such that 
\begin{itemize}
    \item $\abs{v(\pi^{\tau}_i) - v( \pi^*_i )} \leq 2 \ve v_{\max}$ for each $i \in \cala$, and
    \item $f(\pi^{\tau}) \leq f(\pi)$.
\end{itemize}
\end{lemma}

\begin{proof}
Let $\pi^*$ be a given optimal solution. Take $\tau^* \in L^{\cala}$ so that the valuation $v(\pi^*_i)$ is in the interval $[\tau^*(i), \tau^*(i) + \ve v_{\max}]$ for each $i \in \cala$. Note that such $\tau^*$ always exists by Lemma~\ref{limit the range}. Since we apply {\sc LoadBalancing} with $l_i = \tau(i)$ and $u_i = \tau(i) + \ve v_{\max}$ in {\sc MainProcedure} for some $\tau$ with $\tau \sim \tau^*$, we obtain an allocation $\pi^{\tau}$ that corresponds to $\tau$. Then, the inequality $\abs{v(\pi^{\tau}_i) - v( \pi^*_i )} \leq 2 \ve v_{\max}$ holds by reordering the agents appropriately. By the choice of $\pi$ in {\sc MainProcedure}, the inequality $f(\pi^{\tau}) \leq f(\pi)$ holds.
\end{proof}

In preparation for the analysis, we show another bound on the valuation of an agent in an optimal solution. 
Note that a similar result is shown by Alon et al.~\cite{alon1998approximation} for a different problem, 
and our proof for the following lemma is based on their argument.

\begin{lemma}
\label{limit the range2}
Assume that $v_{\max} < \mu$. Let $\pi^*$ be an optimal allocation of goods. Then, 
\[
\f{\mu}{2} < v( \pi^*_i ) < 2 \mu
\]
holds for any $i \in \cala$. 
\end{lemma}

\begin{proof}
The upper bound is obvious because $v( \pi^*_i ) < \mu + v_{\max} < 2 \mu$ by Lemma \ref{limit the range}. 
Assume that there is some agent $i \in \cala$ with $v( \pi^*_i ) \leq \mu / 2$. This implies that $v( \pi^*_i ) < \mu$, and so there must be at least one agent $k$ such that $v( \pi^*_k ) > \mu$. 
We treat the following cases separately.  

Assume that there is a good $j \in \pi^*_k$ that satisfies $v_j < v( \pi^*_k ) - v( \pi^*_i )$. We get a new allocation $\pi'$ from $\pi^*$ by reassigning $j$ to $i$. Then, we obtain 
\begin{align*}
    v( \pi'_i ) v( \pi'_k ) - v( \pi^*_i ) v( \pi^*_k ) &= (v( \pi^*_i ) + v_j)(v( \pi^*_k ) - v_j )- v( \pi^*_i ) v( \pi^*_k ) \\
    &= v_j(v( \pi^*_k ) - v( \pi^*_i ) -v_j  ) \\
    &> 0, 
\end{align*}
which contradicts the optimality of $\pi^*$. 

Assume that for all $j \in \pi^*_k$ the inequality $v_j \geq v( \pi^*_k ) - v( \pi^*_i )$ holds. This implies that $v_j > \mu - \mu/2 = \mu / 2 \geq v( \pi^*_i )$. We also see that $v_j \le v_{\max} < \mu < v( \pi^*_k )$. We get a new allocation $\pi'$ from $\pi^*$ by reassigning whole $\pi^*_i$ to agent $k$, and a single good $j \in \pi^*_k$ to agent $i$. Then, we obtain 
\begin{align*}
    v( \pi'_i ) v( \pi'_k ) - v( \pi^*_i ) v( \pi^*_k ) &= v_j(v( \pi^*_k ) + v( \pi^*_i ) - v_j ) - v( \pi^*_i ) v( \pi^*_k ) \\
    &= (v_j - v( \pi^*_i ) )(v( \pi^*_k )  - v_j  ) \\
    &> 0, 
\end{align*}
which contradicts the optimality of $\pi^*$.
\end{proof}

We are now ready to evaluate the performance of {\sc MainProcedure}. 

\begin{lemma} \label{preadditive}
Assume that $v_{\max} < \mu$ and $0 < \ve \leq 1/5$. 
Let $\pi$ be the allocation returned by {\sc MainProcedure} and let $\OPT$ be the optimal value. 
Then, it holds that
\[
f(\pi) \geq \OPT - 48 \ve v_{\max}.
\]
\end{lemma}

\begin{proof}
By Lemma \ref{neighbor}, there exist an allocation $\pi^{\tau}$ and an optimal solution $\pi^*$ such that 
\begin{equation}
    \abs{v( \pi^{\tau}_i)  - v(\pi^*_i)} \leq 2 \ve v_{\max}, \label{C1}
\end{equation}
and $f(\pi^{\tau}) \leq f(\pi)$.
Let $S = f(\pi^{\tau})$. Since $S \leq f(\pi)$, in order to obtain $f(\pi) \geq \OPT - 48 \ve v_{\max}$, it suffices to show that $\OPT - S \leq 48 \ve v_{\max}$.

We first evaluate the ratio between $\OPT$ and $S$ as follows: 
\begin{align} \label{0}
    \f{\OPT}{S} &= \left( \prod_{i \in \cala} \f{v( \pi^*_i )}{v(\pi_i^{\tau})} \right)^{1/n} \nonumber \\
    &\leq \f{1}{n} \sum_{i}  \f{v( \pi^*_i )}{v(\pi_i^{\tau})} &  &\text{(by AM-GM inequality)} \nonumber \\
    &\leq \f{1}{n} \sum_{i} \left( 1 +  \f{2 \ve v_{\max} }{v(\pi_i^{\tau})} \right)  & & \text{(by (\ref{C1}))} \nonumber \\
    &= 1+ \f{2 \ve v_{\max} }{n} \sum_{i}   \f{ 1 }{v(\pi_i^{\tau})},
\end{align}
where we use the inequality of arithmetic and geometric means (AM-GM inequality) in the first inequality. 
By using (\ref{0}), the difference between $\OPT$ and $S$ can be evaluated as follows:
\begin{align} \label{1}
\OPT - S &= S \left( \f{\OPT}{S} - 1 \right) \nonumber \\
&\leq 2 \ve v_{\max} \left( \f{1}{n} \sum_{i}   \f{ S }{ v(\pi_i^{\tau})}  \right)   & & \text{(by (\ref{0}))} \nonumber \\
&= 2 \ve v_{\max} \left( \f{S }{H}  \right),
\end{align}
where we define $1/H = \f{1}{n}  \sum_{i} 1/ v(\pi_i^{\tau})$, that is, $H$ is the harmonic mean of $v(\pi_i^{\tau})$. 
Therefore, to obtain an upper bound on $\OPT - S$, it suffices to give upper bounds on $S$ and $1/H$. 

We obtain an upper bound on $S$ as follows: 
\begin{align} \label{S}
    S &= \left( \prod_{i} v(\pi_i^{\tau}) \right)^{1/n} \nonumber \\
    &\leq \left( \prod_{i} ( v( \pi^*_i )+ 2 \ve v_{\max} ) \right)^{1/n} & & \text{(by (\ref{C1}))} \nonumber \\
    &\leq \left( \prod_{i} (2 \mu + 2 \ve \mu) \right)^{1/n} & & \text{(by Lemma \ref{limit the range2} and $v_{\max} < \mu$)} \nonumber \\ 
    &= 2 \mu ( 1 +  \ve ). 
\end{align}
Similarly, we obtain an upper bound on $1/H$ as follows: 
\begin{align} \label{1/H}
    \f{1}{H} &= \f{1}{n} \sum_{i}   \f{ 1 }{ v(\pi_i^{\tau})} \nonumber \\
    &\leq \f{1}{n} \sum_{i} \f{ 1 }{ v( \pi^*_i ) - 2 \ve v_{\max} } & & \text{(by (\ref{C1}))} \nonumber \\
    &\leq \f{1}{n}  \sum_{i} \f{ 1 }{ \mu/2 - 2 \ve \mu }  &  & \text{(by Lemma \ref{limit the range2} and $v_{\max} < \mu$)} \nonumber \\
    &= \f{ 2 }{ \mu (1 - 4 \ve)}, 
\end{align}
where we note that $v( \pi^*_i ) - 2 \ve v_{\max} \ge \mu/2 - 2 \ve \mu  > 0$ if $\ve \leq 1/5$. 

Therefore, for $\ve \leq 1/5$, we obtain 
\begin{equation} \label{2}
\f{S}{H} \leq \f{4(1+\ve)}{1 - 4\ve} \leq 24
\end{equation}
by (\ref{S}) and (\ref{1/H}).  
Hence, it holds that $\OPT - S \leq 48 \ve v_{\max}$ by (\ref{1}) and (\ref{2}), which completes the proof. 
\end{proof}

This lemma shows that 
{\sc MainProcedure} is an additive PTAS for {\sc Identical Additive NSW} under the assumption that $v_{\max} < \mu$.  

It is worth noting that {\sc MainProcedure} is not only an additive PTAS, but also a PTAS in the conventional sense. 

\begin{lemma}\label{ratio}
Assume that $v_{\max} < \mu$ and $0 < \ve \leq 1/5$. Let $\pi$ be the allocation returned by {\sc MainProcedure} and let $\OPT$ be the optimal value. 
Then, it holds that
\[
f(\pi) \geq \f{\OPT}{1 + 20 \ve}. 
\]
\end{lemma}

\begin{proof}
Let $S = f(\pi^\tau)$ be the value as in the proof of Lemma~\ref{preadditive}. 
According to inequalities (\ref{0}) and (\ref{1/H}), we obtain
\begin{align*}
    \f{\OPT}{S} &\leq 1 + \f{ 2 \ve v_{\max} }{H} & & \text{(by (\ref{0}))} \\
    &\leq 1 + \f{ 4 \ve v_{\max} }{\mu (1 - 4 \ve)} & & \text{(by (\ref{1/H}))} \\
    &\leq 1 + \f{ 4 \ve  }{1 - 4 \ve}  & & \text{(by $v_{\max} < \mu$)} \\
    &\leq 1 + 20 \ve.  & & \text{(by $0 < \ve \leq 1/5$)}
\end{align*}
Since $f(\pi) \ge S$, this shows that $f(\pi) \geq \OPT / (1 + 20 \ve)$. 
\end{proof}

\subsection{Approximation Performance of {\sc MaxNashWelfare}}
\label{sec:analysispre}

We have already seen in the previous subsection that {\sc MainProcedure} is a PTAS and an additive PTAS for {\sc Identical Additive NSW} under the assumption that $v_{\max} < \mu$. In this subsection, we analyze the effect of {\sc Preprocessing} and show that {\sc MaxNashWelfare} is a PTAS and an additive PTAS. 
As a warm-up, we first show that {\sc MaxNashWelfare} is a PTAS in the conventional sense.

\begin{proposition}
\label{prop:ratio}
Let $I = (\cala, \calg, \mathbf{v})$ be an instance of {\sc Identical Additive NSW} and 
suppose that $0 < \ve \leq 1/5$. 
Let $\pi$ be the allocation returned by {\sc MaxNashWelfare} and let $\OPT$ be the optimal value. 
Then, it holds that
 \[
  f(\pi) \ge \f{\OPT}{1 + 20 \ve}.
 \]
\end{proposition}
 
 \begin{proof}
 Let $\pi^*$ be an optimal allocation. Let $\cala_0$ and $\calg_0$ be the set of agents and goods removed in {\sc Preprocessing} respectively. Set $I' = I \sm (\cala_0,\calg_0)$. In an optimal solution $\pi^*$, for each good $j \in \calg_0$ there exists an agent $i$ that satisfies $\pi^*_i = \{ j \}$  by Lemma \ref{preprocessing}. By rearranging the agents and the goods appropriately, we can assume that $\pi^*_i= \pi_i$ for each $i \in \cala_0$. Then the following holds:
\[
     \f{\OPT}{f(\pi)} = \left( \prod_{i \in \cala} \f{v( \pi^*_i )}{v(\pi_i)} \right)^{1/n}  
     = \left( \prod_{i \in \cala \sm \cala_0} \f{v( \pi^*_i )}{v(\pi_i)} \right)^{1/n}. 
\]
Let $A(I')$ be the objective function value of the solution returned by {\sc MainProcedure} for instance $I'$, and let $\OPT(I')$ be the optimal value of instance $I'$. Set $k = |\cala_0|$. Then, we obtain 
\[
     \left( \prod_{i \in \cala \sm \cala_0} \f{v( \pi^*_i )}{v(\pi_i)} \right)^{1/n} = \left(  \f{\OPT(I')}{A(I')} \right)^{(n-k)/n} 
     \leq (1 + 20 \ve)^{(n-k)/n} 
     \leq 1 + 20 \ve 
\]
by Lemma \ref{ratio}, which completes the proof. 
\end{proof}
 
The proof of Proposition~\ref{prop:ratio} is easy, because the multiplicative error is not amplified when we recover the agents and goods removed in {\sc Preprocessing}, that is, 
${\OPT} / {f(\pi)} \le {\OPT(I')} / {A(I')}$. 
However, this property does not hold when we consider the additive error, which makes the situation harder. 
Nevertheless, we show that the additive error of $f(\pi)$ is bounded by $O(\ve  v_{\max})$ with the aid of the mean-value theorem for differentiable functions.

 \begin{proposition}
 \label{prop:performanceMNW}
 Let $I = (\cala, \calg, \mathbf{v})$ be an instance of {\sc Identical Additive NSW} and suppose that $0 < \ve \leq 1/192$. 
 Let $\pi$ be the allocation returned by {\sc MaxNashWelfare} and let $\OPT$ be the optimal value. 
 Then, it holds that
 \[
 f(\pi) \ge \OPT - 192 \ve v_{\max}.
 \]

 \end{proposition}
 \begin{proof}
Let $A=f(\pi)$ and let $\pi^*$ be an optimal allocation. Let $\cala_0$ and $\calg_0$ be the set of agents and goods removed in {\sc Preprocessing}, respectively.  Set $k = |\cala_0|$. In an optimal solution $\pi^*$, for each good $j \in \calg_0$ there exists an agent $i$ that satisfies $\pi^*_i = \{ j \}$  by Lemma \ref{preprocessing}. By rearranging the agents and the goods appropriately, we can assume that $\pi^*_i = \pi_i$ for each $i \in \cala_0$.
Set $I' = I \sm (\cala_0,\calg_0)$. Let $A(I')$ be the objective function value of the solution returned by {\sc MainProcedure} for instance $I'$, and let $\OPT(I')$ be the optimal value of instance $I'$. 

We define a function $g \colon \R \to \R$ as
  \[
  g(x) = \left( \prod_{j \in \calg_0} v_j \right)^{1/n} x^{(n-k)/n}.
  \]
 By using $g$, the expression to be evaluated can be written as follows:  
  \begin{align}
      \OPT - A = g(\OPT(I')) - g(A(I')). \label{eq:06}
  \end{align}
  Since $g$ is differentiable, by the mean value theorem, there exists a real number $c$ such that 
  \begin{align}
      & A(I') \leq c \leq \OPT(I'),  \label{eq:07} \\
      & g(\OPT(I')) - g(A(I')) = (\OPT(I') - A(I')) g'(c). \label{eq:08}
  \end{align}
  By (\ref{eq:06}), (\ref{eq:08}), and Lemma~\ref{preadditive}, we obtain 
  \begin{align} \label{OPT-A}
   \OPT - A \leq 48 \ve v_{\max}(I') g'(c), 
  \end{align}
  where $v_{\max}(I') = \max_{j \in \calg \sm \calg_0} v_j$.   
  Therefore, all we need to do is to evaluate $g'(c)$. 

  For this purpose, we first give a lower bound on $c$ as follows: 
 \begin{align}
     c &\geq A(I')  & & (\text{by (\ref{eq:07})}) \nonumber\\
     &\geq \OPT(I') - 48 \ve v_{\max}(I') & & \text{(by Lemma~\ref{preadditive})} \nonumber \\
     &= \left( \prod_{i \in \cala \sm \cala_0} v( \pi^*_i ) \right)^{1/(n-k)} - 48 \ve v_{\max}(I') \nonumber \\ 
     &\geq \left( \prod_{i \in \cala \sm \cala_0} \f{\mu(I')}{2} \right)^{1/(n-k)} - 48 \ve v_{\max}(I')  & & \text{(by Lemma \ref{limit the range2})} \nonumber \\
     &\geq \f{\mu(I')}{2} - 48 \ve v_{\max}(I') & & \text{(by $|\cala \sm \cala_0| = n-k$)} \nonumber \\
     &\geq \left( \f{1}{2} - 48 \ve \right) v_{\max}(I').  & & \text{(by $\mu(I') > v_{\max}(I')$)} \label{eq:09}
 \end{align}
 By using this inequality, we obtain the following upper bound on $g'(c)$: 
 \begin{align} \label{f'(c)}
     g'(c) &= \f{n-k}{n} \left( \prod_{j \in \calg_0} v_j \right)^{1/n} \left( \f{1}{c}  \right)^{k/n} \nonumber \\
     &\leq  \left( \f{v_{\max}}{c}  \right)^{k/n} & & \text{(by $|\calg_0| = k$ and $v_j \le v_{\max}$)}  \nonumber \\
     &\leq \left( \f{v_{\max}}{ (1/2 - 48 \ve) v_{\max}(I')}  \right)^{k/n} \nonumber & & \text{(by (\ref{eq:09}))} \\
     &\leq \left( \f{ 4 v_{\max}}{v_{\max}(I')}  \right)^{k/n} \nonumber & & \text{(by $0 < \ve \le 1/192$)} \\
     &\leq \f{ 4 v_{\max}}{v_{\max}(I')}.  & & \text{(by $v_{\max} \geq v_{\max}(I')$)} 
 \end{align}
 Hence, we obtain $\OPT - A \leq 192 \ve v_{\max}$ from (\ref{OPT-A}) and (\ref{f'(c)}), which completes the proof. 
 \end{proof}
 
 By setting $\ve$ appropriately, Theorem \ref{main result} follows from Proposition \ref{prop:performanceMNW}.

\begin{proof}[Proof of Theorem \ref{main result}]
Suppose that we are given an instance of {\sc Identical Additive NSW} and a real value $\ve > 0$. 
Define $\ve'$ as the largest value subject to $1 / \ve'$ is an integer and $\ve' \le \min ( 1/ 192,\  \ve / 192)$. 
That is, $\ve' := 1/ \lceil \max(192, 192/\ve) \rceil$. 
Then, apply {\sc MaxNashWelfare} in which $\ve$ is replaced with $\ve'$. 
Since $0 < \ve' \leq 1/192$,   
{\sc MaxNashWelfare} returns an allocation $\pi$ such that 
$f(\pi) \ge \OPT - 192 \ve' v_{\max} \ge \OPT - \ve v_{\max}$
by Proposition \ref{prop:performanceMNW}. 
As described in Section~\ref{sec:description}, the running time of {\sc MaxNashWelfare} is $(n m/\ve')^{O(1/\ve')}$, which can be rewritten as $(n m/\ve)^{O(1/\ve)}$. 
This completes the proof of Theorem \ref{main result}. 
\end{proof}

%%\section{Concluding Remarks}
%%
%%It is open whether we can obtain an additive FPTAS for the problem. 

%%
%% Bibliography
%%

%% Please use bibtex, 

\bibliography{citation}

\end{document}